\documentclass[runningheads]{llncs}

\usepackage[T1]{fontenc}
\usepackage{amsmath,amssymb,amsfonts}
\usepackage{graphicx}
\usepackage{booktabs}
\usepackage[ruled,vlined]{algorithm2e}
\usepackage{multirow}
\usepackage{threeparttable}
\usepackage{siunitx}
\usepackage[hidelinks]{hyperref}
\usepackage{xcolor,colortbl}
\usepackage{orcidlink}

\definecolor{LightCyan}{rgb}{0.88,1,1}

\begin{document}

% ----- Title page ----------------------------------------------------
\title{When AI Agents Meet MEV:\\Cross-Chain Arbitrage in the Agentic Economy}
\titlerunning{AI Agents and Cross-Chain MEV}

% Anonymized for double-blind review
% \author{Anonymous Authors}
% \authorrunning{Anonymous}
% \institute{Anonymized for review}

\author{
    Wei Ye\inst{1}\orcidlink{0009-0002-6632-377X}
    \and
    Jingyan Xu\inst{2}\orcidlink{0009-0006-0864-6035}
    \and
    Yuanhong Wu\inst{2}\orcidlink{0009-0000-3226-9547}
}
\authorrunning{W. Ye et al.}
\institute{
    Department of Economics, Fordham University, New York, NY, USA\\
    \email{wei.ye@fordham.edu}
    \and
    Department of Computer and Information Science, Fordham University, New York, NY, USA\\
    \email{\{jxu246,ywu463\}@fordham.edu}
}

\maketitle

% ----- Abstract ------------------------------------------------------
\begin{abstract}
We study cross-chain arbitrage when autonomous AI agents, rather than
humans or bots, are the searchers. We model agents as both arbitrage
extractors and Maximal Extractable Value targets, derive the
optimal trade size for a risk-averse agent under mean-variance utility
with stochastic bridge delays, and formalize multi-chain path selection
as a belief-weighted online learning problem whose belief estimates
converge under a Robbins--Monro schedule. Using 23{,}000 Uniswap~V3
swap events across Ethereum, Arbitrum, and Base, we find that
Ethereum--Arbitrum price gaps average 0.044\% at 10-second resolution
and Arbitrum--Base gaps average 0.013\%, so \$10{,}000 trades clear in
63\% of L2--L2 windows via CCTP while L1--L2 routes require \$50{,}000
or more for comparable viability. Our adaptive path-selection algorithm
outperforms standard baselines by 11\% on average, and moderate
randomization cuts MEV exposure by over 50\%
with only modest profit loss.

\end{abstract}

% =====================================================================
\section{Introduction}
\label{sec:intro}

By early 2026, autonomous AI agents had completed over 140 million
on-chain payments totaling \$43 million in USDC, transacting through
infrastructure such as Coinbase's Agentic Wallets and Circle's
Nanopayments. These figures mark the emergence of an \emph{agentic
economy}, in which AI agents act as independent economic participants
that initiate transactions and execute financial strategies without
human oversight.

Yet the intersection of agentic AI and Maximal Extractable Value (MEV)
remains unexplored. MEV research models human or bot searchers on a
single chain and ignores the behavioral signatures of autonomous
agents. Work on agentic payments, in turn, emphasizes portfolio
management and yield optimization but overlooks the adversarial
execution environment in which every pending transaction is visible
and exploitable. The features that make agents effective
searchers---speed, cross-chain awareness, and continuous
monitoring---also expose their algorithmic regularities to sandwich
attacks and front-running. Agents therefore operate simultaneously as
MEV extractors and targets, a duality that existing cross-chain MEV
models obscure further by abstracting away the transaction costs that
often render arbitrage infeasible.

This paper bridges agentic AI and cross-chain MEV with five
contributions.
   \textbf{Agent-aware MEV framework:} We model cross-chain MEV
    with AI agents as economic participants, treating them as both
    searchers and targets.
  \textbf{Optimal arbitrage under uncertainty:} We derive the
    optimal trade size for a risk-averse agent under mean-variance
    utility, accounting for stochastic bridge delays, price volatility,
    and transaction costs.
  \textbf{Adaptive path selection:} We formalize multi-chain
    path selection as a belief-weighted online learning problem and
    establish its convergence.
  \textbf{Vulnerability analysis:} We characterize how
    deterministic agent behavior amplifies susceptibility to SDA
    attacks and propose mitigations.
  \textbf{Empirical evaluation:} Using DeFiLlama price, gas, and
    bridge data, we calibrate cross-chain price distributions, backtest
    the adaptive algorithm against baselines, and quantify how
    transaction costs shrink the set of viable arbitrage opportunities.

Three modeling choices are specific to the actors being software agents
rather than humans or ad-hoc bots. First, the risk-aversion parameter
$\lambda$ in~(7) is a \emph{design} parameter chosen by the developer,
not an estimated behavioral trait, so the comparative statics
of~(9) are prescriptive: they tell a designer how trade size should
scale with configured risk tolerance. Second, agents follow policies
that are (near-)deterministic functions of public on-chain state;
Section~4 shows that this makes the searcher itself a predictable MEV
target---a threat model absent from the human-trader MEV
literature---and Section~4.3 quantifies the resulting
protection--profitability tradeoff, which only arises for algorithmic
actors that can commit to randomization. Third, continuous multi-chain
operation motivates online belief updating rather than
one-shot optimization. The searcher-side tools are deliberately
standard; the contribution is the closed loop between searcher
optimization and target-side exposure.

\section{Related Work}\label{sec:litreview}
\textbf{MEV measurement and mitigation.}
Daian et al.~\cite{daian2020flash} formalized MEV.
Qin et al.~\cite{qin2022quantifying} quantify on-chain arbitrage,
liquidation, and sandwich extraction. Wahrst\"atter et
al.~\cite{wahrstatter2023time} and Heimbach et
al.~\cite{heimbach2023ethereum} document builder concentration under
proposer-builder separation. Heimbach and
Wattenhofer~\cite{heimbach2022eliminating} construct a sandwich-resistant
ordering rule with linear latency overhead;
Mancino and Rezzoli~\cite{mancino2025sandwiched} show that 2{,}932
sandwich attacks occurred over two months even on private routing.

\textbf{Cross-chain arbitrage and MEV.}
Obadia et al.~\cite{obadia2021unity} formalize cross-domain MEV.
\"Oz et al.~\cite{oz2025pandora} introduce the inventory-based vs.\
bridge-based (SIA/SDA) distinction; their
follow-up~\cite{oz2025cross} measures 242{,}535 cross-chain arbitrages
worth \$868.64M across nine chains, with SIA settling in 9~s versus
242~s for SDA. Li et al.~\cite{li2025walls} identify cross-chain
sandwich attacks extracting \$5.27M in two months. Heimbach et
al.~\cite{heimbach2024non} analyze non-atomic execution risk;
Milionis et al.~\cite{milionis2024automated} bound arbitrage profits
in AMMs with fees; Li et al.~\cite{li2023mev} characterize MEV under
greedy sequencing.

\textbf{Autonomous agents on blockchain.}
Alqithami~\cite{alqithami2026autonomous} surveys 3{,}000+ agent-blockchain
integrations and identifies four architectural patterns.
Xu~\cite{xu2026agent} proposes a layered economy for autonomous agents.
Li et al.~\cite{li2026a402} extend the x402 payment standard with atomic
service channels. Marino and Juels~\cite{marino2025giving} catalog
attack vectors when agents hold on-chain wallets.

\textbf{This work.}
Existing cross-chain MEV models assume human or bot searchers and
abstract away transaction costs; existing agentic-blockchain work treats
the execution environment as benign. We close both gaps: we model
autonomous agents as risk-averse cross-chain searchers with explicit
gas and bridge costs, prove convergence of an online path-selection
rule, and characterize agents' dual exposure as MEV targets.

% =====================================================================
\section{Cross-Chain Arbitrage Model}
\label{sec:model-sia}

\subsection{System Model and Problem Setup}

Consider a set of $N$ blockchains
$\mathcal{B} = \{B_1, B_2, \ldots, B_N\}$, each hosting one or more
DEXs. An autonomous AI agent $\alpha$ monitors token prices across all
chains and seeks to extract profit through cross-chain arbitrage. For
a given token pair $(A, B)$, let $p_i(t)$ denote the price of token
$A$ in terms of token $B$ on blockchain $B_i$ at time $t$. Throughout, we fix a single token pair $(A,B)$ and suppress it from
the notation; $p_i(t)$ thus abbreviates $p_i^{A/B}(t)$.

The agent $\alpha$ is characterized by the following:
A \emph{wallet state} $w_i(t)$ representing the agent's token
    holdings on each chain $B_i$;
A \emph{risk aversion parameter} $\lambda > 0$ governing the
    agent's tolerance for price uncertainty;
  A \emph{belief vector}
    $\hat{\boldsymbol{\rho}}(t) = (\hat{\rho}_{ij}(t))$ representing the
    agent's estimated success probability for each cross-chain path
    $(i, j)$.

When the agent executes arbitrage between blockchains $B_i$ and $B_j$,
it exchanges $x$ units of token $B$ for token $A$ on chain $B_i$, then
sells token $A$ on chain $B_j$ after bridging. This operation incurs
two types of costs:
   Gas fees for on-chain swaps, denoted $c_{\text{gas}}(t)$;
  Bridging fees for cross-chain asset transfer, denoted
    $c_{\text{bridge}}(t + \tau_{ij})$, where $\tau_{ij}$ is the delay
    for bridging from $B_i$ to $B_j$.

The total transaction cost is:
\begin{equation}
\label{eq:cost}
  c(t) = c_{\text{gas}}(t) + c_{\text{bridge}}(t + \tau_{ij}).
\end{equation}

The agent's single-period profit from arbitraging $x$ units of token
$B$ from chain $B_i$ to chain $B_j$ is:
\begin{equation}
\label{eq:profit}
  \pi_{ij}(x, t) = \frac{x}{p_i(t)} \cdot \mathbb{E}[p_j(t + \tau_{ij})] - x - c(t),
\end{equation}
where $x / p_i(t)$ is the quantity of token $A$ obtained on chain
$B_i$, and $\mathbb{E}[p_j(t + \tau_{ij})]$ is the expected price of
token $A$ on chain $B_j$ at the time the bridged assets arrive.

In a multi-period setting, the agent's decision at each time $t$ is
whether to execute arbitrage or wait. The value
function satisfies:
\begin{equation}
\label{eq:bellman}
V_t(s_t) = \max \begin{cases}
\displaystyle\max_{(i,j) \in \mathcal{P}} \left[ \pi_{ij}(x, t) + \delta \, \mathbb{E}[V_{t+\tau_{ij}}] \right], \\[6pt]
\delta \, \mathbb{E}[V_{t+1}(s_{t+1})]
\end{cases}
\end{equation}
where $\mathcal{P} \subseteq \{(i,j) \mid i \neq j\}$ is the set of
valid cross-chain paths, $\delta \in (0,1)$ is the discount factor,
and $s_t$ encodes the market state.

\subsection{Optimal Trade Size under Uncertainty}

The agent cannot observe $p_j(t + \tau_{ij})$ at the time of execution
due to the bridging delay. Assume the future price follows:
\begin{equation}
\label{eq:price-dist}
  p_j(t + \tau_{ij}) \sim \mathcal{N}(\mu_j, \sigma_j^2),
\end{equation}
where $\mu_j$ and $\sigma_j^2$ are the mean and variance of the
destination chain price, which the agent estimates from historical
data.

Substituting into~(\ref{eq:profit}), the expected profit is:
\begin{equation}
\label{eq:expected-profit}
  \mathbb{E}[\pi] = \frac{x}{p_i(t)} \cdot \mu_j - x - c(t),
\end{equation}
and the variance of profit is:
\begin{equation}
\label{eq:var-profit}
  \mathrm{Var}(\pi) = \left(\frac{x}{p_i(t)}\right)^2 \cdot \sigma_j^2.
\end{equation}

A risk-averse agent maximizes mean-variance utility:
\begin{equation}
\label{eq:utility}
  U(x) = \mathbb{E}[\pi] - \frac{\lambda}{2} \cdot \mathrm{Var}(\pi)
       = \frac{x \mu_j}{p_i(t)} - x - c(t) - \frac{\lambda}{2} \cdot \frac{x^2 \sigma_j^2}{p_i(t)^2}.
\end{equation}

Taking the first-order condition:
\begin{equation}
\label{eq:foc}
  \frac{\partial U}{\partial x} = \frac{\mu_j}{p_i(t)} - 1 - \frac{\lambda \sigma_j^2 x}{p_i(t)^2} = 0.
\end{equation}

Solving for $x$ yields the optimal trade size:
\begin{equation}
\label{eq:optimal-x}
  x^* = \frac{p_i(t)\bigl(\mu_j - p_i(t)\bigr)}{\lambda \sigma_j^2}.
\end{equation}

The optimal trade size $x^*$ is increasing in the expected price gap
$(\mu_j - p_i(t))$ and decreasing in both the risk aversion $\lambda$
and the price volatility $\sigma_j^2$. In the limiting case
$\lambda \to 0$, $x^* \to \infty$, reflecting that
a risk-neutral agent would trade without bound as long as
$\mu_j > p_i(t)$. For AI agents, $\lambda$ can be interpreted as a
configurable parameter that governs how aggressively the agent trades.

\subsubsection{Arbitrage Threshold Condition.}

Arbitrage is profitable only if $\pi_{ij}(x,t) > 0$.
Rearranging~(\ref{eq:profit}) yields a necessary condition:
\begin{equation}
\label{eq:threshold}
  \mathbb{E}[p_j(t+\tau_{ij})] > p_i(t) \left(1 + \frac{c(t)}{x}\right).
\end{equation}

This condition reveals a fundamental frictions--arbitrage tradeoff.
Higher transaction costs $c(t)$ or smaller trade sizes $x$ require a
larger price gap to justify arbitrage. For AI agents operating with
limited wallet balances, this threshold may exclude many opportunities
that are viable for well-capitalized searchers. Conversely, agents
that can aggregate capital across chains or batch transactions can
lower the effective $c(t)/x$ ratio.

% \begin{remark}
% The normality assumption yields a tractable closed-form solution. In
% practice, cross-chain price gaps are right-skewed, better approximated
% by a log-normal distribution. We adopt the log-normal model in our
% empirical evaluation and show that the
% qualitative conclusions of the threshold
% condition~(\ref{eq:threshold}) are robust to this distributional
% choice.
% \end{remark}

\subsection{Adaptive Path Selection}

In a multi-chain environment, the agent must choose which path
$(i, j)$ to use at each time step. This is an online decision problem:
the agent does not know in advance which paths are reliable, and must
learn from experience while continuing to exploit its current
knowledge.

\subsubsection{Belief Update Mechanism.}

The agent maintains a belief $\hat{\rho}_{ij}(t) \in [0, 1]$ for each
path $(i,j)$, representing the estimated probability that arbitrage on
that path will succeed (i.e., yield positive profit). Let
$\eta \in (0,1)$ be the learning rate. After each attempt, beliefs are
updated as follows. If the arbitrage succeeds (positive profit):
\begin{equation}
\label{eq:belief-success}
  \hat{\rho}_{ij}(t+1) = \hat{\rho}_{ij}(t) + \eta\bigl(1 - \hat{\rho}_{ij}(t)\bigr).
\end{equation}
If the arbitrage fails (negative profit or execution failure):
\begin{equation}
\label{eq:belief-fail}
  \hat{\rho}_{ij}(t+1) = \hat{\rho}_{ij}(t) - \eta \hat{\rho}_{ij}(t)
                      = (1 - \eta)\hat{\rho}_{ij}(t).
\end{equation}

The agent selects the path that maximizes belief-weighted expected
profit:
\begin{equation}
\label{eq:path-select}
  (i^*, j^*)_t = \arg\max_{(i,j) \in \mathcal{P}}
    \left\{ \hat{\rho}_{ij}(t) \cdot \pi_{ij}(x^*, t) + \delta \cdot \mathbb{E}[V_{t+1}(s_{t+1})] \right\}.
\end{equation}

In practice, the continuation value $\mathbb{E}[V_{t+1}]$ is difficult
to compute exactly. A myopic approximation drops the continuation
term:
\begin{equation}
\label{eq:path-select-myopic}
  (i^*, j^*)_t = \arg\max_{(i,j) \in \mathcal{P}}
    \left\{ \hat{\rho}_{ij}(t) \cdot \pi_{ij}(x^*, t) \right\}.
\end{equation}

Algorithm~1 summarizes the complete adaptive path selection procedure.

% \begin{figure}[h]
% \centering
% \begin{minipage}{0.95\linewidth}
% \hrule\vspace{4pt}
% \textbf{Algorithm 1:} Adaptive Cross-Chain Path Selection \vspace{2pt}
% \hrule\vspace{4pt}
% \textbf{Input:} Chains $\mathcal{B}$, paths $\mathcal{P}$, learning rate $\eta$, risk aversion $\lambda$, exploration rate $\epsilon$, time horizon $T$ \\
% \textbf{Output:} Cumulative profit $\Pi$
% \vspace{2pt}

% \textbf{Initialize:} $\hat{\rho}_{ij} \leftarrow 0.5$ for all $(i,j) \in \mathcal{P}$; $\Pi \leftarrow 0$

% \textbf{for} $t = 1, 2, \ldots, T$ \textbf{do}

% \hspace{1em} Observe prices $\{p_i(t)\}_{i=1}^N$ and costs $\{c_{ij}(t)\}$

% \hspace{1em} Compute $x^*_{ij}$ and $\pi_{ij}(x^*_{ij}, t)$ for all $(i,j) \in \mathcal{P}$ via~(\ref{eq:optimal-x}) and~(\ref{eq:profit})

% \hspace{1em} \textbf{with probability} $\epsilon$: select $(i^*, j^*)$ uniformly at random from $\mathcal{P}$

% \hspace{1em} \textbf{otherwise:} $(i^*, j^*) \leftarrow \arg\max_{(i,j)} \hat{\rho}_{ij}(t) \cdot \pi_{ij}(x^*_{ij}, t)$

% \hspace{1em} Execute arbitrage on path $(i^*, j^*)$ with trade size $x^*_{i^*j^*}$

% \hspace{1em} Observe realized profit $\pi^{\text{real}}$

% \hspace{1em} \textbf{if} $\pi^{\text{real}} > 0$ \textbf{then}

% \hspace{2em} $\hat{\rho}_{i^*j^*} \leftarrow \hat{\rho}_{i^*j^*} + \eta(1 - \hat{\rho}_{i^*j^*})$

% \hspace{1em} \textbf{else}

% \hspace{2em} $\hat{\rho}_{i^*j^*} \leftarrow (1 - \eta) \cdot \hat{\rho}_{i^*j^*}$

% \hspace{1em} $\Pi \leftarrow \Pi + \pi^{\text{real}}$

% \textbf{return} $\Pi$
% \vspace{4pt}\hrule
% \end{minipage}
% \label{alg:adaptive}
% \end{figure}

\begin{algorithm}[t]
\SetAlgoLined
\KwIn{Chains $\mathcal{B}$, paths $\mathcal{P}$, learning rate $\eta$, risk aversion $\lambda$, exploration rate $\epsilon$, horizon $T$}
\KwOut{Cumulative profit $\Pi$}
Initialize $\hat{\rho}_{ij} \leftarrow 0.5$ for all $(i,j) \in \mathcal{P}$; $\Pi \leftarrow 0$\;
\For{$t = 1, \ldots, T$}{
  Observe $\{p_i(t)\}, \{c_{ij}(t)\}$; compute $x^*_{ij}$, $\pi_{ij}$ via~(\ref{eq:optimal-x}),~(\ref{eq:profit})\;
  \eIf{$\mathrm{rand}() < \epsilon$}{
    Select $(i^*,j^*)$ uniformly at random from $\mathcal{P}$\;
  }{
    $(i^*,j^*) \leftarrow \arg\max_{(i,j)} \hat{\rho}_{ij}(t) \cdot \pi_{ij}(x^*_{ij},t)$\;
  }
  Execute on $(i^*,j^*)$ with size $x^*_{i^*j^*}$; observe $\pi^{\text{real}}$\;
  \lIf{$\pi^{\text{real}} > 0$}{$\hat{\rho}_{i^*j^*} \leftarrow \hat{\rho}_{i^*j^*} + \eta(1 - \hat{\rho}_{i^*j^*})$}
  \lElse{$\hat{\rho}_{i^*j^*} \leftarrow (1 - \eta) \hat{\rho}_{i^*j^*}$}
  $\Pi \leftarrow \Pi + \pi^{\text{real}}$\;
}
\Return{$\Pi$}\;
\caption{Adaptive Cross-Chain Path Selection}
\end{algorithm}

The algorithm uses an $\epsilon$-greedy exploration strategy: with
probability $\epsilon$ the agent selects a random path to gather
information about underexplored routes, and otherwise exploits its
current beliefs. This balances exploration and exploitation, connecting
to the classical multi-armed bandit framework~\cite{auer2002finite}.

\begin{proposition}
Under the belief update rules~(\ref{eq:belief-success})--(\ref{eq:belief-fail}),
if the true success probability $\rho_{ij}$ of path $(i,j)$ is
stationary, and the path is selected infinitely often, then
$\hat{\rho}_{ij}(t) \to \rho_{ij}$ as $t \to \infty$.
\end{proposition}

\begin{proof}
Let $Y_t \in \{0, 1\}$ be the outcome indicator at time $t$ (1 for
success, 0 for failure), with $\mathbb{E}[Y_t] = \rho_{ij}$. The
update rule can be written as:
\begin{equation}
  \hat{\rho}_{ij}(t+1) = (1 - \eta)\hat{\rho}_{ij}(t) + \eta Y_t.
\end{equation}
This is an exponential moving average with smoothing parameter $\eta$.
Taking expectations in steady state,
$\mathbb{E}[\hat{\rho}_{ij}] = (1-\eta)\mathbb{E}[\hat{\rho}_{ij}] + \eta \rho_{ij}$,
which gives $\mathbb{E}[\hat{\rho}_{ij}] = \rho_{ij}$. The variance of
the estimator is
$\mathrm{Var}(\hat{\rho}_{ij}) = \frac{\eta}{2-\eta} \rho_{ij}(1-\rho_{ij})$,
which decreases as $\eta \to 0$. Thus, with a decaying learning rate
schedule $\eta_t \to 0$ satisfying $\sum \eta_t = \infty$ and
$\sum \eta_t^2 < \infty$, the estimator converges almost surely to
$\rho_{ij}$ by the Robbins--Monro theorem. \qed
\end{proof}

\section{Agent Vulnerability to MEV Attacks}
\label{sec:vulnerability}
We now examine the agent as a target rather than a searcher of MEV.

\subsection{Agent Predictability and Exploitation}

Traditional MEV targets are human traders whose behavior contains
noise from emotional decisions, variable timing, and inconsistent
trade sizes. AI agents, by contrast, follow deterministic or
near-deterministic policies. An agent that checks prices every
$\Delta t$ seconds, trades whenever the price gap exceeds a threshold
$\theta$, and uses a fixed trade size $x^*$ from~(\ref{eq:optimal-x})
generates a transaction pattern that is observable and learnable by an
adversary.

Let $\mathcal{H}_t = \{(t_k, x_k, i_k, j_k)\}_{k=1}^{t}$ denote the
history of the agent's transactions up to time $t$, where each entry
records the timestamp, trade size, source chain, and destination
chain. An MEV attacker observing $\mathcal{H}_t$ can attempt to
predict the agent's next transaction
$(t_{t+1}, x_{t+1}, i_{t+1}, j_{t+1})$. Define the attacker's
prediction accuracy as:
\begin{equation}
\label{eq:pred-accuracy}
  \mathcal{A}(t) = \Pr\left[\hat{t}_{t+1} \approx t_{t+1}, \; \hat{i}_{t+1} = i_{t+1}\right],
\end{equation}
where $\hat{t}_{t+1}$ and $\hat{i}_{t+1}$ are the attacker's
predictions. For a fully deterministic agent, $\mathcal{A}(t) \to 1$
as $t$ grows, since the attacker can learn the agent's policy exactly
from observed history.

\subsubsection{Sandwich Attack on Cross-Chain Agents.}

Suppose the agent submits a swap of size $x_\alpha$ on a DEX,
exchanging token $B$ for token $A$ at initial price $p_0$. An attacker
who predicts this transaction sandwiches it: a front-run buy of $x_a$
units of token $A$ pushes the price to
$p_1 = p_0 + \Delta p_{\text{front}}$, the agent's swap then executes
at the inflated $p_1$, and a back-run sell of $x_a$ units at
$p_2 = p_1 + \Delta p_{\text{agent}}$ closes the position.

The attacker's profit is:
\begin{equation}
\label{eq:sandwich-profit}
  \pi_{\text{sand}} = x_a (p_2 - p_0) - c_{\text{gas}} - b,
\end{equation}
where $b$ is the bribe or priority fee paid to secure the desired
transaction ordering.

The agent's loss from being sandwiched is the difference between the
price it would have received without the attack and the actual
execution price:
\begin{equation}
\label{eq:agent-loss}
  L_\alpha = x_\alpha \cdot \Delta p_{\text{front}} / p_0.
\end{equation}

For agents executing cross-chain arbitrage, this loss directly reduces
the arbitrage profit $\pi_{ij}$ from~(\ref{eq:profit}), potentially
turning a profitable opportunity into a net loss.

\subsubsection{Cross-Chain Predictive Front-Running.}

In a cross-chain setting, an additional attack vector arises. When an
agent initiates a bridge transfer on chain $B_i$, this transaction is
visible on-chain before the assets arrive on chain $B_j$. An attacker
monitoring chain $B_i$ can detect the bridge transaction, infer the
agent's intended trade on chain $B_j$, and pre-position on chain $B_j$
before the bridged assets arrive.

Let $\tau_{ij}$ be the bridge delay. The attacker's information
advantage window is exactly $\tau_{ij}$: the attacker knows the
agent's intent on chain $B_j$ a full $\tau_{ij}$ seconds before the
agent can act. The attacker's expected profit from predictive
front-running is:
\begin{equation}
\label{eq:predictive-profit}
  \pi_{\text{pred}} = x_a \cdot \mathbb{E}[\Delta p \mid \text{bridge detected}] - c_{\text{pos}},
\end{equation}
where $c_{\text{pos}}$ is the cost of pre-positioning on chain $B_j$.

This attack is more severe for bridge-based arbitrage
(\cite{oz2025cross} report averages of 242 seconds of latency) than
for inventory-based arbitrage (averaging 9 seconds), as longer delays
give the attacker more time to act.

\subsection{Mitigation Strategies}

We propose three strategies that an AI agent can employ to reduce its
MEV exposure.

\subsubsection{Transaction Timing Randomization.}

Instead of executing at fixed intervals $\Delta t$, the agent adds
random noise to its execution timing:
\begin{equation}
\label{eq:random-timing}
  t_{\text{exec}} = t_{\text{planned}} + \xi, \quad \xi \sim \mathrm{Uniform}(-\delta_t, \delta_t),
\end{equation}
where $\delta_t$ controls the noise magnitude. This reduces the
attacker's ability to predict when the agent will transact. However,
large $\delta_t$ may cause the agent to miss time-sensitive arbitrage
windows, creating a tradeoff between protection and profitability.

\subsubsection{Trade Size Perturbation.}

The agent perturbs its trade size around the optimal $x^*$:
\begin{equation}
\label{eq:random-size}
  x_{\text{exec}} = x^* + \zeta, \quad \zeta \sim \mathcal{N}(0, \sigma_x^2),
\end{equation}
subject to $x_{\text{exec}} > 0$. Since the attacker's sandwich profit
in~(\ref{eq:sandwich-profit}) depends on correctly anticipating the
agent's trade size to calibrate $x_a$, perturbation reduces the
attacker's expected gain. The agent's utility loss from deviating from
$x^*$ is bounded by:
\begin{equation}
\label{eq:util-loss}
  U(x^*) - \mathbb{E}[U(x^* + \zeta)] = \frac{\lambda \sigma_j^2}{2 p_i(t)^2} \cdot \sigma_x^2,
\end{equation}
which follows from the quadratic form of the utility
function~(\ref{eq:utility}). The agent can choose $\sigma_x^2$ to
balance protection against utility loss.

\subsubsection{MEV-Aware Bridge Selection.}

The agent can incorporate MEV risk into its path selection by
modifying the decision rule in~(\ref{eq:path-select-myopic}). Define
$m_{ij}(t)$ as the estimated MEV exposure on path $(i,j)$, based on
historical sandwich frequency or the bridge's information leakage
properties. The MEV-aware path selection rule becomes:
\begin{equation}
\label{eq:mev-aware}
  (i^*, j^*)_t = \arg\max_{(i,j) \in \mathcal{P}}
    \left\{ \hat{\rho}_{ij}(t) \cdot \pi_{ij}(x^*, t) - m_{ij}(t) \right\}.
\end{equation}

This rule penalizes paths with high MEV exposure. An agent using this
rule may choose a path with a slightly lower expected profit but
significantly lower attack risk. In practice, $m_{ij}(t)$ can be
estimated from on-chain data by counting sandwich events on each DEX
or measuring the frequency of front-running on specific bridge
protocols.

\subsection{Protection-Profitability Tradeoff}

The three mitigation strategies above share a common structure: each
introduces noise or constraints that reduce the agent's predictability
at the cost of some expected profit. We can formalize this tradeoff.
Let $\phi \in [0,1]$ represent the agent's overall protection level,
where $\phi = 0$ is a fully deterministic agent and $\phi = 1$ is a
maximally randomized agent. The agent's effective utility is:
\begin{equation}
\label{eq:tradeoff}
  U_{\text{eff}}(\phi) = (1 - \phi) \cdot U(x^*) + \phi \cdot \bigl(U(x^*) - L_{\text{noise}}(\phi)\bigr) + \phi \cdot S(\phi),
\end{equation}
where $L_{\text{noise}}(\phi)$ is the utility loss from randomization
and $S(\phi)$ is the savings from avoided MEV attacks, which increases
with $\phi$. The optimal protection level $\phi^*$ balances these two
forces.

In competitive MEV environments where multiple attackers monitor agent
behavior, the marginal benefit of protection $S'(\phi)$ is high,
favoring aggressive randomization. In low-competition environments,
the agent may prefer minimal protection to maximize raw arbitrage
profit.

% =====================================================================
\section{Empirical Evaluation}
\label{sec:experiments}

We evaluate our models using two complementary data sources: (i) six
months of aggregate market data from DeFiLlama for parameter
calibration, and (ii) 23{,}000 tick-level Uniswap~V3 swap events
across Ethereum, Arbitrum, and Base spanning five trading days for
direct observation of cross-chain price gaps. We first describe the
data and calibration, then characterize real cross-chain price gaps,
and finally present three simulation experiments: arbitrage threshold
analysis, adaptive path selection, and agent predictability.

\subsection{Data and Calibration}

We collect four datasets from DeFiLlama covering the period September
2025 to March 2026 (six months): Daily OHLCV prices for WETH/USDC, WBTC/USDC, ARB/USDC, and OP/USDC; Daily gas fees for Ethereum, Arbitrum, Base, and Optimism;  Daily bridge volume for Across, Stargate, and Hop Protocol; Daily DEX volume for 30 chains.

In addition, we collect 23{,}000 individual Uniswap~V3 swap events
for the WETH/USDC pair across three chains: Ethereum (3{,}000 swaps
across the 0.01\%, 0.05\%, and 0.30\% fee-tier pools on March 23,
2026), Arbitrum (10{,}000 swaps across 0.01\% and 0.05\% pools), and
Base (10{,}000 swaps across 0.01\% and 0.05\% pools). The L2 data
spans five days across two periods (February 22--25 and March 23,
2026), with timestamps at second-level resolution. This dataset
enables direct observation of cross-chain price gaps at the
granularity relevant to arbitrage execution and allows us to verify
consistency across multiple trading days.

A key empirical challenge is that cross-chain price discrepancies at
the daily level are negligible, as arbitrage bots close gaps within
milliseconds. Following~\cite{oz2025cross}, we adopt a calibrated
simulation approach: we extract distributional parameters from real
data and use them to generate realistic arbitrage scenarios.

We estimate daily price volatility using the Parkinson
estimator~\cite{parkinson1980extreme}, which uses intraday high--low
range:
\begin{equation}
  \hat{\sigma}^2 = \frac{1}{4\ln 2} \cdot \mathbb{E}\left[\left(\ln \frac{H}{L}\right)^2\right],
\end{equation}
where $H$ and $L$ are the daily high and low prices.
Table~\ref{tab:calibration} reports the calibrated parameters.

\begin{table}[t]
\centering
\caption{Calibrated parameters from DeFiLlama data.}
\label{tab:calibration}
\begin{tabular}{@{}lrr@{\hspace{2em}}lrr@{}}
\toprule
\multicolumn{3}{c}{\textbf{Panel A: Token Volatility}} & \multicolumn{3}{c}{\textbf{Panel B: Chain Fees (USD/day)}} \\
Token & $\bar{p}$ & $\hat{\sigma}$ & Chain & Mean & Median \\
\midrule
WETH/USDC & 3{,}019  & 0.0374 & Ethereum & 668{,}354 & 438{,}810 \\
WBTC/USDC & 90{,}039 & 0.0263 & Base     & 201{,}367 & 136{,}604 \\
ARB/USDC  & 0.22     & 0.0542 & Arbitrum & 45{,}049  & 23{,}123  \\
OP/USDC   & 0.32     & 0.0581 & Optimism & 4{,}844   & 2{,}334   \\
\midrule
\multicolumn{3}{c}{\textbf{Panel C: Bridge Volume (USD)}} & \multicolumn{3}{c}{\textbf{Panel D: DEX Volume (USD)}} \\
\multicolumn{3}{c}{Across: 47{,}491{,}067} & \multicolumn{3}{c}{Ethereum: 2{,}235{,}832{,}499} \\
\multicolumn{3}{c}{Stargate: 36{,}518{,}438} & \multicolumn{3}{c}{Base: 1{,}121{,}527{,}914} \\
\multicolumn{3}{c}{Hop: 4{,}023} & \multicolumn{3}{c}{Arbitrum: 581{,}000{,}984} \\
 & & & \multicolumn{3}{c}{Optimism: 34{,}905{,}833} \\
\bottomrule
\end{tabular}
\end{table}

Two observations from Table~\ref{tab:calibration} are worth noting.
First, L2 native tokens (ARB, OP) exhibit nearly twice the daily
volatility of major assets (WETH, WBTC), implying larger intraday
price gaps and more frequent arbitrage windows. Second, gas fees
differ by over two orders of magnitude across chains: Ethereum
averages \$668K/day in total fees while Optimism averages only \$4.8K,
directly affecting the threshold condition in~(\ref{eq:threshold}).

For the simulation experiments, we model instantaneous cross-chain
price gaps using a log-normal distribution:
\begin{equation}
  \Delta p \sim \mathrm{LogNormal}(\ln(0.03), 1.2),
\end{equation}
with a median gap of 0.03\% of price, calibrated from the Uniswap~V3
swap data (Section~\ref{sec:empirical-gaps}). This right-skewed
specification captures the empirical pattern where the median
Ethereum--Arbitrum gap is 0.031\% but the 90th percentile reaches
0.086\% and extreme values exceed 0.4\%.

\subsection{Empirical Cross-Chain Price Gaps}
\label{sec:empirical-gaps}

Before turning to simulation experiments, we use the Uniswap~V3 swap
data to directly characterize cross-chain price gaps. We aggregate
swap-level execution prices into 10-second bins per chain, compute the
mean execution price per bin, and measure the absolute percentage
difference between chains for bins where both chains have at least one
swap. This yields 112 simultaneous observations for the
Ethereum--Arbitrum pair (limited to the 28-minute overlap window on
March 23) and 954 for Arbitrum--Base (spanning five days of trading).

Table~\ref{tab:crosschain-gaps} reports the cross-chain gap
statistics. The Ethereum--Arbitrum pair exhibits a mean gap of 0.044\%
(median 0.031\%), with the 90th percentile at 0.086\% and extreme
values reaching 0.479\%. These gaps are economically significant: at a
\$50{,}000 trade size with realistic costs (\$5.25 including Ethereum
gas and bridge fees), 80\% of 10-second windows present a profitable
arbitrage opportunity. For the Arbitrum--Base pair, the five-day
dataset reveals a mean gap of 0.013\% (median 0.010\%), consistent
across all five observed days (daily means range from 0.010\% to
0.015\%). Despite these smaller gaps, the dramatically lower L2
execution costs (\$0.70--\$1.20 per trade) mean that at \$10{,}000
trade size via CCTP, 63\% of windows are viable.

\begin{table}[h]
\centering
\caption{Cross-chain price gap statistics (\%, WETH/USDC).}
\label{tab:crosschain-gaps}
\setlength{\tabcolsep}{4pt}
\begin{tabular}{lcrrrr}
\toprule
\textbf{Pair} & $N$ & Med. & Mean & P90 & Max \\
\midrule
\multicolumn{6}{l}{\textit{Cross-chain}} \\
ETH--ARB  & 112 & .031 & .044 & .086 & .479 \\
ARB--Base & 954 & .010 & .013 & .027 & .077 \\
\midrule
\multicolumn{6}{l}{\textit{Intra-chain (different fee tier)}} \\
ETH .01\%--.05\% & 267 & .040 & .049 & .088 & .564 \\
ARB .01\%--.05\% & 400 & .030 & .027 & .041 & .089 \\
\bottomrule
\end{tabular}
\end{table}

Figure~\ref{fig:crosschain-gaps} visualizes the results. Panel~(a)
shows execution prices across all observation periods, covering a
\$1{,}850--\$2{,}150 price range across five trading days. Panel~(b)
compares the gap distributions via box plots: the ETH--ARB pair
(L1--L2, $N=112$) exhibits wider gaps with heavier tails than
ARB--Base (L2--L2, $N=954$). Panel~(c) shows the cumulative
distribution function, highlighting that the median ETH--ARB gap
(0.031\%) is roughly three times the ARB--Base median (0.010\%).
Panel~(d) quantifies arbitrage viability as a function of trade size
under realistic cost assumptions: for ETH--ARB at \$5.25 total cost,
trades below \$5K are rarely viable, whereas ARB--Base via CCTP at
\$0.70 reaches 63\% viability at \$10K.

\begin{figure}[t]
  \centering
  \includegraphics[width=0.78\linewidth]{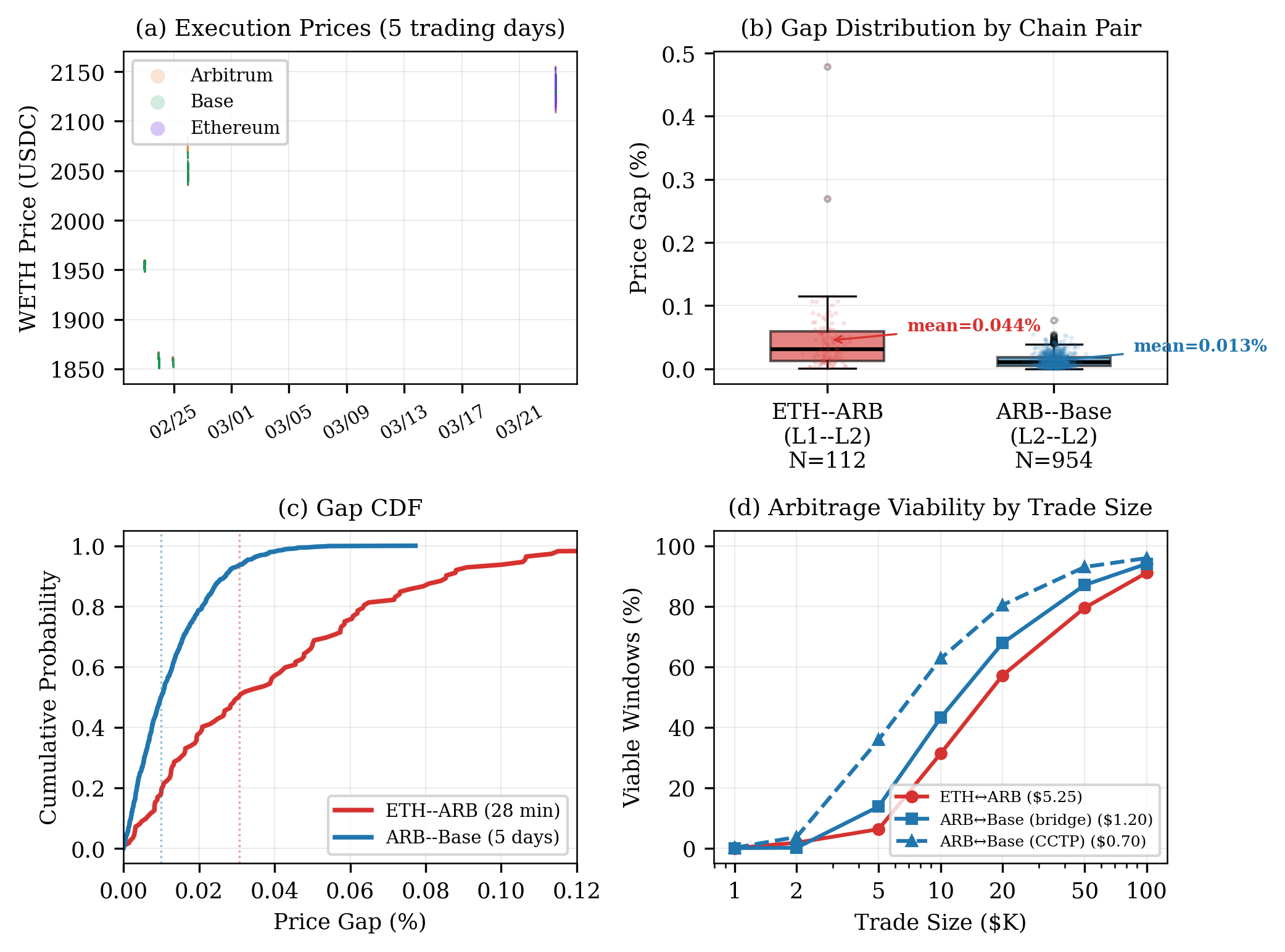}
  \caption{Cross-chain price gaps from Uniswap~V3 swap events.
    (a)~Execution prices across chains.
    (b)~Gap distribution by chain pair.
    (c)~Cumulative distribution of gaps.
    (d)~Fraction of viable arbitrage windows by trade size and cost
    scenario.}
  \label{fig:crosschain-gaps}
\end{figure}

Two findings from this analysis are noteworthy. First, the
Ethereum--Arbitrum cross-chain gap (mean 0.044\%) is comparable in
magnitude to the intra-chain gap between different fee tiers on
Ethereum (mean 0.049\%), suggesting that cross-chain price dispersion
is driven more by execution timing and information latency than by
fundamental market segmentation. Second, the L2--L2 gap
(Arbitrum--Base, mean 0.013\%) is roughly one-third of the L1--L2 gap,
consistent with the shorter bridge delays between L2 chains reported
by~\cite{oz2025cross}. Importantly, the ARB--Base gap is stable across
all five observed trading days (daily means: 0.010\%, 0.013\%,
0.013\%, 0.014\%, 0.015\%), spanning a period in which ETH prices
ranged from \$1{,}850 to \$2{,}150. This consistency across market
conditions strengthens the case that cross-chain price gaps are a
structural feature of multi-chain DEX liquidity rather than a
transient artifact. These observations validate the model's prediction
in~(\ref{eq:threshold}) that cross-chain arbitrage viability is
determined by the gap-to-cost ratio rather than the absolute gap size.

\subsection{Experiment 1: Arbitrage Threshold Analysis}

This experiment evaluates the threshold condition
from~(\ref{eq:threshold}) under realistic cost parameters. We generate
50{,}000 price gap samples from the calibrated log-normal distribution
and compute the fraction of gaps that exceed the cost threshold
$c(t)/x$ for each chain and trade size.

Figure~\ref{fig:threshold} presents the results. Panel~(a) shows that
on Ethereum, trade sizes below \$25{,}000 yield almost no viable
opportunities due to high gas costs, while on Optimism and Base, even
\$1{,}000 trades produce viable arbitrage in 30--40\% of simulated
gaps. Panel~(b) compares the log-normal gap distribution (median
0.03\%, calibrated from swap data) against a naive normal assumption,
showing that the log-normal model produces a heavier right tail.
Panel~(c) reports a sensitivity analysis: we vary the assumed median
gap from 0.01\% to 0.5\%, marking the empirically observed median
(0.03\%), and find that the qualitative ranking across chains is
stable. L2 chains remain viable at 10--100$\times$ smaller trade sizes
than Ethereum regardless of the gap assumption.

\begin{figure}[t]
  \centering
  \includegraphics[width=0.78\linewidth]{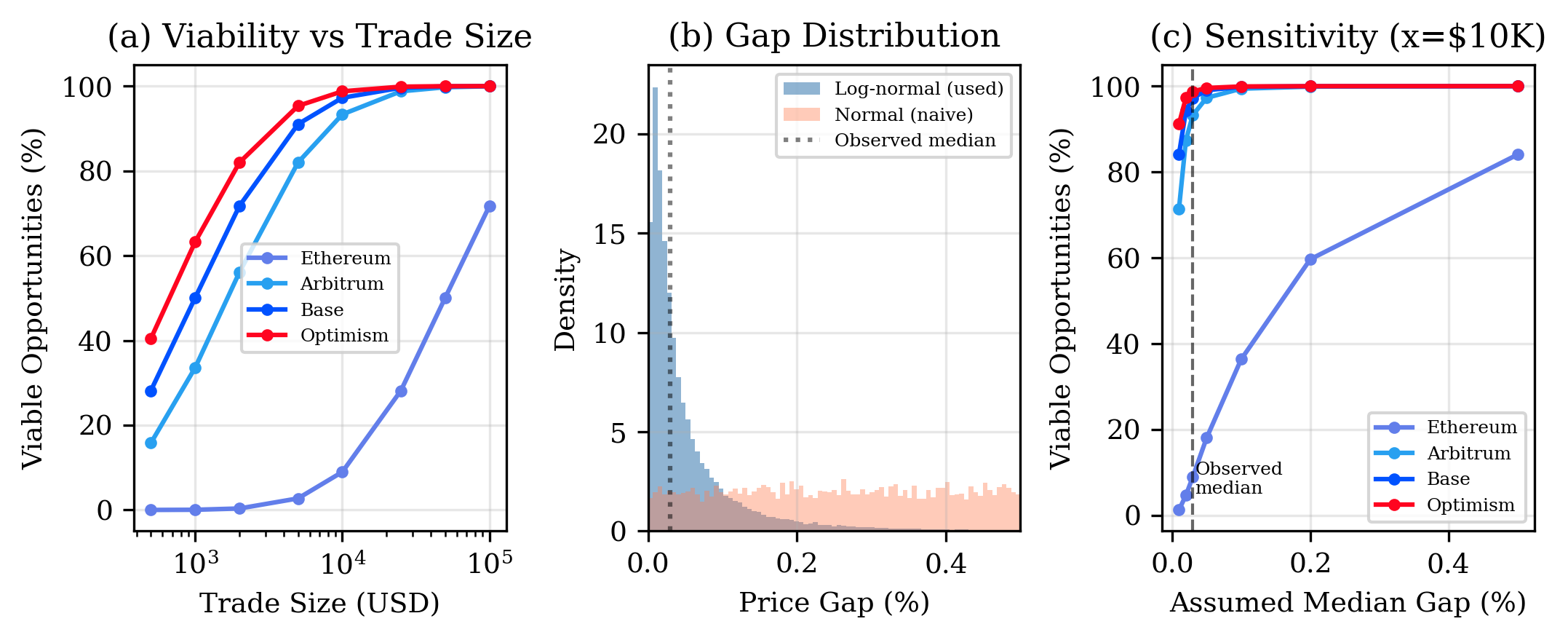}
  \caption{Arbitrage threshold analysis.
    (a)~Fraction of viable opportunities by trade size and chain.
    (b)~Log-normal vs.\ normal gap distribution.
    (c)~Sensitivity to median gap assumption at \$10K trade size.}
  \label{fig:threshold}
\end{figure}

This result has a direct implication for agent design: agents with
limited capital should prioritize L2 chains where the
cost-to-trade-size ratio is favorable, consistent with the path
selection framework in Section~\ref{sec:model-sia}.

\subsection{Experiment 2: Adaptive Path Selection}

We simulate 180 days of cross-chain arbitrage across four chains (12
directed paths) and compare four strategies: \textbf{Random}: selects a path uniformly at random each day; \textbf{Greedy}: selects the path with the highest estimated net profit based on current price gaps and costs; \textbf{$\epsilon$-Greedy}: tracks empirical mean reward per path, exploits the best with probability $1-\epsilon$ and explores randomly with probability $\epsilon = 0.1$; \textbf{Adaptive (ours)}: the belief-weighted algorithm from
    Algorithm~1 with $\eta = 0.1$ and $\epsilon = 0.1$.

Daily returns are computed using real volatility data from the
WETH/USDC OHLCV series and calibrated per-chain gas costs. Each path
has a latent success probability drawn from a structured distribution
where lower-cost paths tend to have higher success rates, reflecting
the empirical finding that L2 paths are more
reliable~\cite{oz2025cross}. To ensure robustness, we run 100
independent trials, each with a different draw of latent success
probabilities.

\begin{table}[t]
\centering
\caption{Strategy performance (100 runs, 180 days).}
\label{tab:strategy}
\begin{tabular}{lrrr}
\toprule
Strategy & Mean Profit & Std & 95\% CI \\
\midrule
Random            & $-$135 & 345 & [$-$764, 566] \\
Greedy            & 2{,}367 & 697 & [944, 3{,}604] \\
$\epsilon$-Greedy & 2{,}242 & 586 & [1{,}094, 3{,}458] \\
Adaptive (ours)   & \textbf{2{,}627} & 560 & [1{,}738, 3{,}706] \\
\bottomrule
\end{tabular}
\end{table}

Table~\ref{tab:strategy} reports the results. The Adaptive strategy
achieves the highest mean profit (\$2{,}627), outperforming Greedy in
57 of 100 runs and $\epsilon$-Greedy in 69 of 100 runs. Notably, the
Random strategy loses money on average ($-$\$135), confirming that
naive path selection is insufficient when cross-chain gaps are small
relative to costs. The Adaptive strategy's advantage is most
pronounced when latent path success probabilities are heterogeneous,
as the belief-update mechanism allows the agent to identify and
concentrate on high-quality paths. When success probabilities are
similar across paths, all informed strategies perform comparably,
explaining why Adaptive does not dominate in every run.

Figure~\ref{fig:pathselection}(a) shows cumulative profit curves with
interquartile ranges. The Adaptive strategy not only achieves higher
terminal profit but also exhibits lower variance, reflecting the
stabilizing effect of belief-weighted decisions.
Figure~\ref{fig:pathselection}(b) reports the sensitivity of the
Adaptive strategy to the learning rate $\eta$. Performance is stable
across $\eta \in [0.05, 0.3]$, with degradation at extreme values:
$\eta = 0.01$ learns too slowly, while $\eta = 0.5$ overreacts to
individual outcomes.

\begin{figure}[t]
  \centering
  \includegraphics[width=0.78\linewidth]{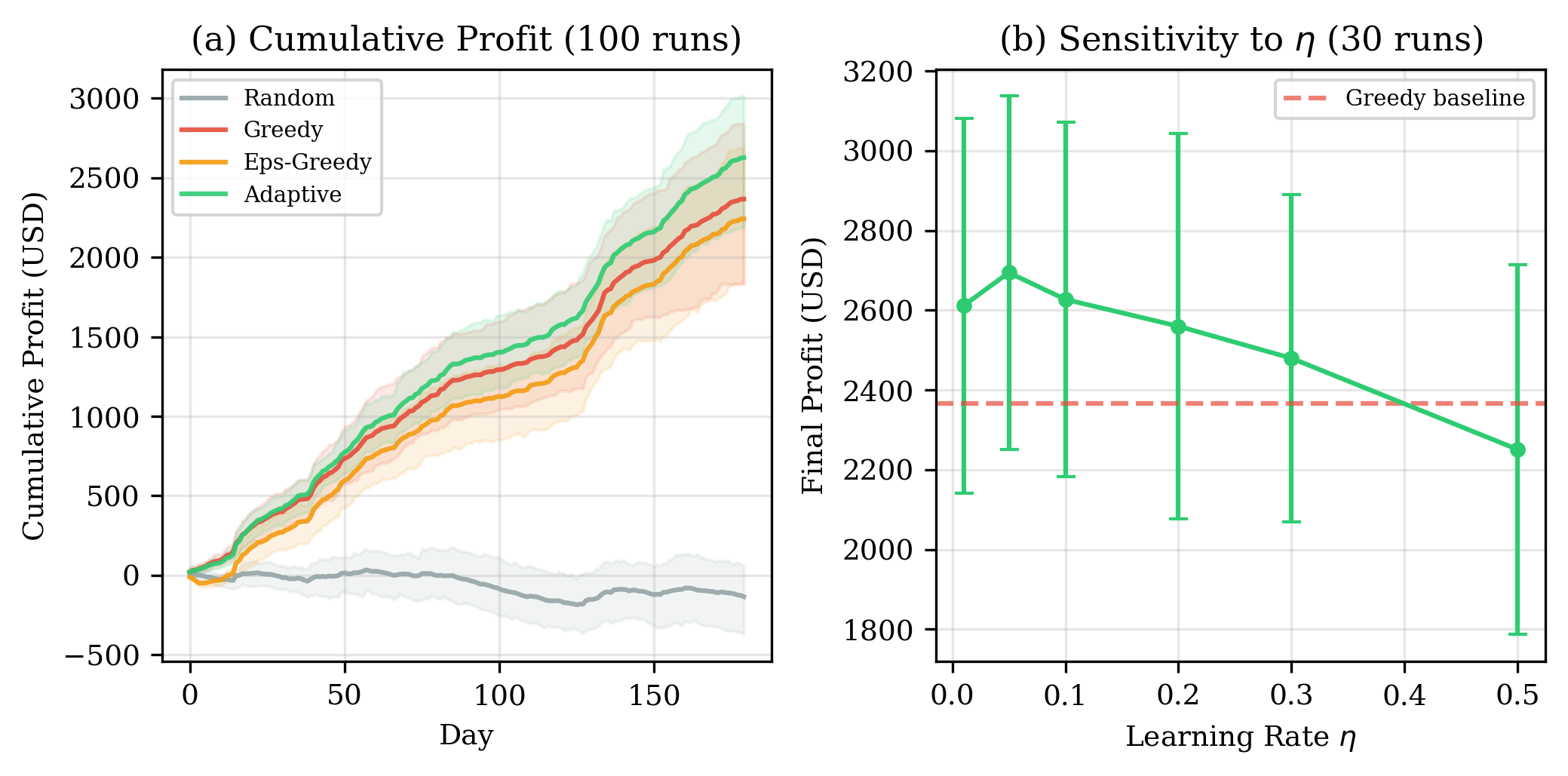}
  \caption{Adaptive path selection.
    (a)~Cumulative profit by strategy (100 runs, shaded region = IQR).
    (b)~Sensitivity to learning rate $\eta$ (error bars = 1~std).}
  \label{fig:pathselection}
\end{figure}

\subsection{Experiment 3: Agent Predictability and Protection}

This experiment evaluates the vulnerability analysis from
Section~\ref{sec:vulnerability}. We simulate an agent executing
periodic trades with a base strategy (trade every $\sim$10 periods,
size $\sim$\$10{,}000), and vary the protection level
$\phi \in \{0, 0.2, 0.4, 0.6, 0.8, 1.0\}$ by adding increasing noise
to timing, trade size, and path selection. An attacker observes the
agent's last 20 trades and predicts the next trade's timing and path
using median-based and majority-vote heuristics. We run 30 independent
trials for each $\phi$ level.

Table~\ref{tab:predictability} reports the attacker's prediction
accuracy. At $\phi = 0$ (fully deterministic agent), the attacker
achieves 100\% joint prediction accuracy. Introducing moderate
randomization ($\phi = 0.4$) reduces joint accuracy to 54.2\%, and at
$\phi = 0.8$ it drops to 16.0\%. Path prediction is harder to disrupt
than timing prediction, as the agent's preferred paths are more stable
than its execution timing.

\begin{table}[t]
\centering
\setlength{\abovecaptionskip}{2pt}
\setlength{\belowcaptionskip}{2pt}
\caption{Attacker prediction accuracy vs.\ protection level.}
\label{tab:predictability}
\setlength{\tabcolsep}{6pt}
\renewcommand{\arraystretch}{0.95}
\begin{tabular}{lrrr}
\toprule
$\phi$ & Joint (\%) & Timing (\%) & Path (\%) \\
\midrule
0.0 & 100.0 & 100.0 & 100.0 \\
0.2 & 88.0  & 96.1  & 91.6  \\
0.4 & 54.2  & 64.0  & 83.5  \\
0.6 & 29.3  & 42.9  & 68.3  \\
0.8 & 16.0  & 31.6  & 52.5  \\
1.0 & 7.6   & 25.6  & 31.4  \\
\bottomrule
\end{tabular}
\end{table}
Figure~\ref{fig:predictability}(a) visualizes the smooth degradation
of attacker accuracy across all three metrics.
Figure~\ref{fig:predictability}(b) presents the
protection-profitability tradeoff predicted by~(\ref{eq:tradeoff}): as
$\phi$ increases, MEV protection improves (red line) but arbitrage
profit decreases (green line). The curves cross near
$\phi \approx 0.4$--$0.6$, suggesting that moderate randomization
offers the best balance between protection and profitability for a
risk-averse agent.

\begin{figure}[t]
  \centering
  \includegraphics[width=0.78\linewidth]{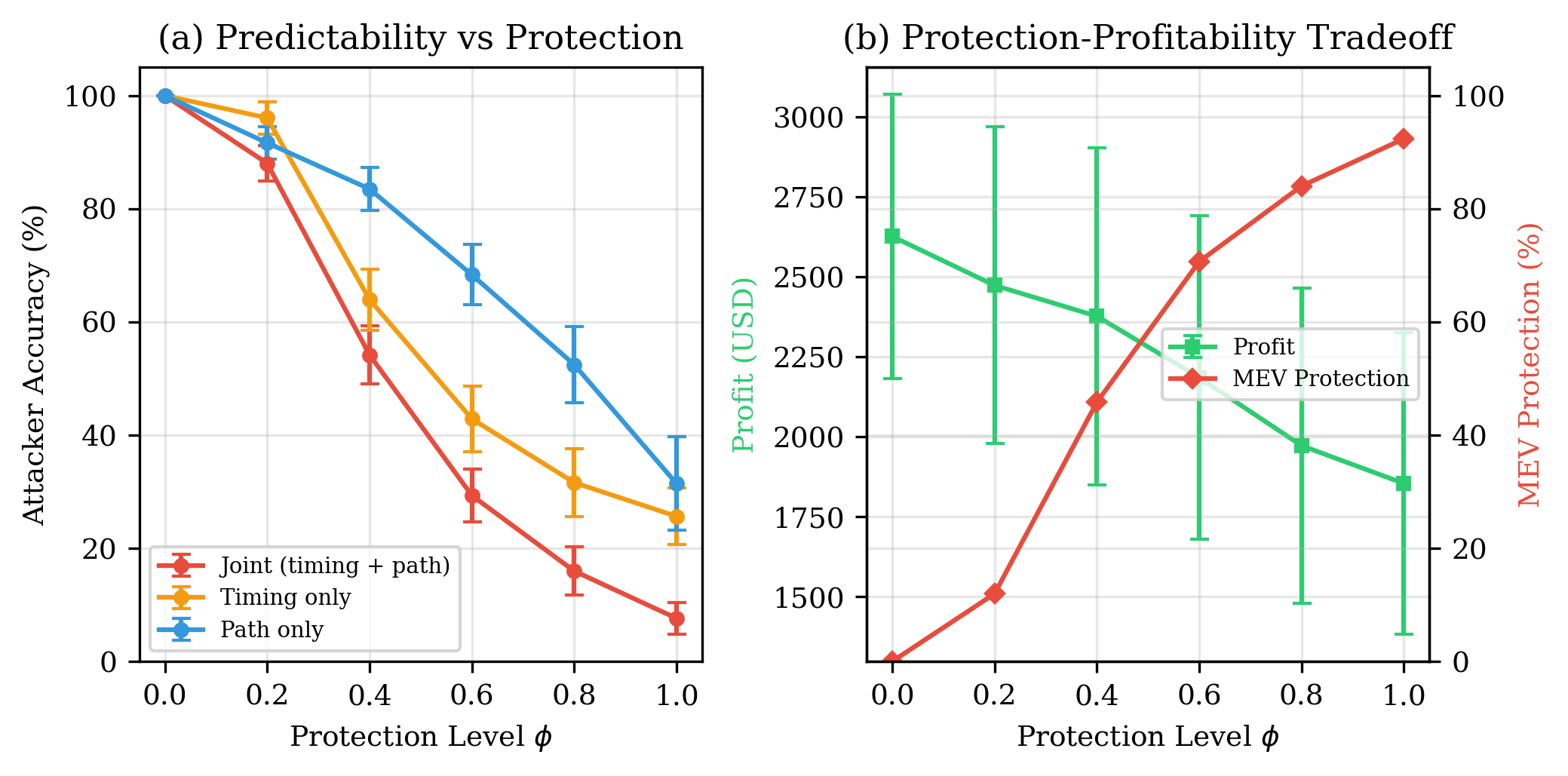}
  \caption{Agent predictability.
    (a)~Attacker accuracy vs.\ protection level (30 runs, error bars =
    1~std).
    (b)~Protection-profitability tradeoff: profit (left axis) and MEV
    protection (right axis) as functions of $\phi$.}
  \label{fig:predictability}
\end{figure}

% =====================================================================
\section{Conclusion}
\label{sec:conclusion}

This paper formalizes the intersection of autonomous AI agents and
cross-chain MEV, treating agents as both arbitrage searchers and MEV
targets. We validate the framework using 23{,}000 tick-level
Uniswap~V3 swap events across Ethereum, Arbitrum, and Base over five
trading days, supplemented by simulations calibrated to DeFiLlama
market data.

Three findings stand out. Transaction costs dominate cross-chain
arbitrage viability: Ethereum--Arbitrum price gaps average 0.044\% at
10-second resolution while Arbitrum--Base gaps average 0.013\%
($\sigma = 0.002\%$), so once gas costs are netted, \$10{,}000 trades
clear in 63\% of L2--L2 windows via CCTP, whereas Ethereum--Arbitrum
routes require \$50{,}000 or more for comparable viability. The
adaptive path selection algorithm outperforms standard baselines by
11\% on average---most strongly when path quality is
heterogeneous---and is robust to learning rates
$\eta \in [0.05, 0.3]$. Moderate randomization
($\phi \approx 0.4$--$0.6$) cuts MEV exposure by over 50\% with only a
modest profit loss, indicating that designers should calibrate
unpredictability to the competitive intensity of the environment
rather than maximize it.

Future work includes scaling L1--L2 data collection to weeks for formal distributional estimation, and extending the framework to multi-agent strategic settings with AMM-specific price impact~\cite{alqithami2026autonomous,marino2025giving}.

% =====================================================================
%  Bibliography
% =====================================================================
{\footnotesize
\bibliographystyle{splncs04}
\bibliography{mybibliography}
}

\end{document}